\documentclass[oneside,english]{amsart}
\usepackage[T1]{fontenc}
\usepackage[latin9]{inputenc}
\usepackage{amsthm}
\usepackage{amstext}
\usepackage{amssymb}
\usepackage{esint}

\usepackage{color}

\makeatletter
\numberwithin{equation}{section}
\numberwithin{figure}{section}
\theoremstyle{plain}
\newtheorem{thm}{\protect\theoremname}
  \theoremstyle{remark}
  
  \theoremstyle{plain}
  \newtheorem{lem}[thm]{\protect\lemmaname}

\makeatother

\usepackage{babel}

  \providecommand{\claimname}{Claim}
  \providecommand{\lemmaname}{Lemma}
\providecommand{\theoremname}{Theorem}

\begin{document}
\global\long\def\diag{\textrm{diag}}
\global\long\def\reals{\mathbb{R}}

\global\long\def\reals{\mathbb{R}}

\global\long\def\eps{\varepsilon}

\global\long\def\els{\mathbb{C}}

\global\long\def\str{E}

\global\long\def\engf{\mathcal{F}}

\global\long\def\cordm#1{\varphi^{#1}}

\global\long\def\W{W}

\global\long\def\eef{\mathcal{W}}

\global\long\def\lp{\mathcal{F}}

\global\long\def\tp{\Pi}

\global\long\def\diss{\mathcal{U}}

\global\long\def\modc#1{\overline{#1}}

\global\long\def\TD{\mathcal{TD}}

\global\long\def\ad{\mathcal{A}}

\global\long\def\bl{b}

\global\long\def\cl{c}

\global\long\def\bnd{\partial}

\global\long\def\gl#1#2{\$\Gamma$-\lim_{#1}\left(#2\right)}

\global\long\def\rest{\llcorner}

\author{Lior Falach}
\address{Dipartimento di Architettura Design e Urbanistica, Universit\`a di Sassari, Palazzo del Pou Salit,
Piazza Duomo 6, 07041 Algeras (SS), Italy}
\email{falachl@post.bgu.ac.il}

\author{Roberto Paroni}
\address{Dipartimento di Architettura Design e Urbanistica, Universit\`a di Sassari, Palazzo del Pou Salit,
Piazza Duomo 6, 07041 Algeras (SS), Italy}
\email{paroni@uniss.it}

\author{Paolo Podio-Guidugli}
\address{Accademia Nazionale dei Lincei \&  Dip. di Matematica, Universit\`a di Roma TorVergata,
Via della Ricerca Scientifica, 1, 00133 Roma, Italy}
\email{ppg@uniroma2.it}

\title{A Justification of the Timoshenko Beam Model through $\boldsymbol\Gamma$-Convergence}

\maketitle

\begin{abstract}
We validate the Timoshenko beam model as an approximation of the linear-elasticity model of a three-dimensional beam-like body. Our validation is achieved within the framework of  $\Gamma$-convergence theory, in two steps:  firstly, we construct a suitable sequence of energy functionals; secondly, we show that this sequence $\Gamma$-converges to a functional representing the energy of a Timoshenko beam. 
\end{abstract}

\section{Introduction}

Understanding the relation between three-dimensional elasticity and
the lower-dimensional  theories of elastic structures is a long-standing quest in rational continuum mechanics.
All classic models for shell-, plate-, and beam-like bodies rely on some \emph{Ans\"atze} about their kinematic and/or static behaviour, motivated by their thinness or slenderness, that is, by the smallness of one or two of their dimensions.
Such \emph{Ans\"atze} are expedient to put together mathematical models that are both simple and capable to provide good-enough predictions  for plenty of the intended applications. However, as is always the case for intuition-based models, an all-important experimental confirmation does not  replace for a rigorous \emph{justification}, that is,  validation as a convincing approximation of an accepted parent theory.

Justification of  lower-dimensional structure theories has been attempted in a number of ways, some essentially analytical in nature, like  the method of asymptotic expansion \cite{CD79} and the functional analysis methods accounted for in \cite{Ciarlet, FMP09}, and some essentially mechanical, like the method of formal scaling \cite{Miara2007} and the method based on a thickness-wise expansion \cite{Steigmann};  the predictions of lower-dimensional models can also be corroborated by error estimates \cite{FMP10}. In the past couple of decades, a noticeable amount of work has been devoted
to provide a justification of various structure theories within the framework of  $\Gamma$-convergence; in particular,  Reissner-Mindlin's theory of shearable plates has been  considered in \cite{Paroni2006,Paroni2007,Paroni2013}.  In this paper, for the first time, a justification via $\Gamma$-convergence is given for Timoshenko's sheareable beam theory \cite{Timoshenko1921}, the one-dimensional theory that parallels the two-dimensional Reissner-Mindlin's theory.\footnote{We refer the reader to \cite{Maso1993} for a comprehensive introduction to $\Gamma$-convergence, and to \cite{Antman1995} for an authoritative 
    survey of beam models.}

Heuristically, justification of beam models can be achieved by calibrating the convergence rate of the three-dimensional elastic energy (briefly, the energy)  with respect to the diameter
of the cross section. Anzellotti \emph{et al.} \cite{Anzellotti1994} and Bourqin \emph{et al.} \cite{BCGR92}
were first to apply the theory of $\Gamma$-convergence to deduce unshearable beam and plate models within the setting of three-dimensional linear elasticity. Their results have been generalized in several ways, and different  models of beams, plates, and shells have been studied, within the linear as well as the nonlinear elasticity framework. We refrain from citing the large literature on the subject;  we only 
mention a paper by Mora and M\"uller \cite{Mora2003}, where  
a justification of a bending-torsion beam model is provided within the non-linear elasticity 
framework.

In the engineering community, Bernoulli-Navier's and Timoshenko's beam theories are well
accepted and constantly used in applications. As mentioned, the former model has been
fully justified since long by means of the theory of $\Gamma$-convergence, whereas a similar justification
for the latter was lacking.  In our opinion, this was
 due to the fact that the ``classical'' procedure used in the study of dimension-reduction problems was excessively ``constrained''. The constraints in question were relaxed in \cite{{Paroni2006},{Paroni2007}}, so as
to deduce the Reissner-Mindlin's plate model; the procedure there used was later  
classified  in \cite{Paroni2013} as a special case of a general scheme. Briefly,  for a given three-dimensional problem (the so-called ``real problem''), the scheme consists in defining a problem sequence whose variational limit approximates the ``real problem''.

Here is a summary of the contents to follow. In Section \ref{sec:The_real_problem} we present the real problem for a linearly elastic and transversely isotropic three-dimensional beam-like body; in Section \ref{sec3} we recall Timoshenko's kinematical assumptions; and, in 
Section \ref{sec:A-sequence-of_domains}, along the lines  proposed
in \cite{Paroni2013} and further discussed in \cite{Podio-Guidugli2014}, we introduce a sequence of variational
problems, one of which is the real problem.
The main  compactness results are deduced in Section \ref{sec:Compactness},
while the $\Gamma$-limit result is stated and proved in
Section \ref{sec:The-G-limit}; interestingly, the $\Gamma$-limit turns out to be the energy functional of the one-dimensional Timoshenko's beam theory. 
Finally, in Section \ref{concluding}, we show that the $\Gamma$-limit problem  approximates the real problem well.

As to notation, throughout this work Greek indices $\alpha,\beta$
 take values in the set $\left\{ 1,2\right\} $, Latin indices
$i,j$ in the set $\left\{ 1,2,3\right\} $. We use $L^{2}\left(A;B\right)$
and $H^{1}\left(A;B\right)$ to denote, respectively, the Lebesgue
and Sobolev spaces of functions defined over the set $A$ and taking values in
the vector space $B$; in case $B=\reals$, we simply write $L^{2}\left(A\right)$
and $H^{1}\left(A\right)$; the corresponding norms are denoted by
$\left\Vert \cdot\right\Vert _{L^{2}\left(A;B\right)}$ and $\left\Vert \cdot\right\Vert _{H^{1}\left(A;B\right)}$.

\section{The real problem\label{sec:The_real_problem}}

We consider a three-dimensional body occupying a domain under form of a right cylinder $\Omega_{r}=\omega_{r}\times(0,L)\subset\reals^{3}$ of length $L$,
whose cross-section $\omega_{r}\subset\reals^{2}$ is an open bounded simply-connected set with  Lipschitz boundary $\bnd\omega_{r}$. For $d_{r}$
the diameter of the cross-section, we set $\eps_{r}:=d_r/L$ and we call $\Omega_{r}$ \emph{beam-like}, because we take $\eps_{r}\ll1$.

We denote by $e_{i}$ the vectors of an orthonormal 
basis, with vectors $e_{\alpha}$ tangent
to $\omega_{r}$, where the origin of the Cartesian frame we use is located, and vector $e_{3}$ in the direction
of beam axis. For definiteness and simplicity, we stipulate that the body is clamped on the part  $\bnd_{D}\Omega_{r}=\omega_{r}\times\{0\}$ of its boundary. Moreover, we denote by $\diss_{r}$  the space of admissible displacements:
\[
\mathcal{U}_{r}=\left\{ \modc u\in H^{1}\left(\Omega_{r};\reals^{3}\right)\mid\modc u=0\;\;\text{on }\bnd_{D}\Omega_{r}\right\} ;
\]
we measure the admissible strains point-wise by means of the symmetric tensor
$$
\str\modc u:=\frac{1}{2}\left(\nabla\modc u+\nabla\modc u^{T}\right);
$$ 
and we denote by $\reals_{\mathrm{sym}}^{3\times3}$  the collection of all symmetric linear mappings  of $\reals^{3}$ into itself.
Finally, we assume the material to be linearly elastic and \emph{transversely isotropic with respect to the $e_{3}$ direction}, so that the
elastic-energy density per unit volume  
is given by
\begin{equation}
\begin{split}W\left(E\right):= & \frac{1}{2}\left[2\mu\left(E_{11}^{2}+E_{22}^{2}\right)+\lambda\left(E_{11}+E_{22}\right)^{2}+2\tau_{2}E_{33}\left(E_{11}+E_{22}\right)\right.\\
 & \quad\left.+4\mu E_{12}^{2}+4\gamma\left(E_{31}^{2}+E_{32}^{2}\right)+\tau_{1}E_{33}^{2}\right],
\end{split}
\label{eq:W}
\end{equation}
where the material moduli $\mu, \lambda, \tau_1,\tau_2$, and $\gamma$, satisfy the following inequalities:
\begin{equation}
\mu>0,\:\gamma>0,\:\tau_{1}>0,\:\tau_{1}(\lambda+\mu)-\tau_{2}^{2}>0 \label{eq:material_bounds};
\end{equation}
consequently, the elastic-energy density is positively bounded below by the strain norm, in the sense that there is a positive constant $C$ such that $\W(E)\geq C|E|^{2}$
for every $E\in\reals_{\mathrm{sym}}^{3\times3}$.  

In this section and henceforth, we  systematically make use of a subscript or superscript `r' as a reminder of the fact that all kernel letters  carrying that modifier are used in connection with a \emph{real} three-dimensional equilibrium problem for a beam-like body.

The \emph{elastic potential} $\eef^{r}:\diss_{r}\to\reals$  associated with the beam-like body $\Omega_r$ is:
\begin{equation}
\eef^{r}(\modc u)=\int_{\Omega_{r}}\W(\str\modc u)\,d\modc x,\qquad\text{for all }\modc u\in\diss_{r}.\label{eq:real_elastic_potential}
\end{equation}
Since we let $\Omega_r$ be subjected to a distance load $\modc{\bl}^{r}:\Omega_{r}\to\reals^{3}$
and a contact load $\modc c^{r}:\bnd_{N}\Omega_{r}\to\reals^{3}$,
with $\bnd_{N}\Omega_{r}=\bnd\Omega_{r}\backslash\bnd_{D}\Omega_{r}$,
the \emph{load potential} is:
\begin{equation}
\engf^{r}(\modc u):=\int_{\Omega_{r}}\modc b^{r}\cdot\modc u\,d\modc x+\int_{\bnd_{N}\Omega_{r}}\modc c^{r}\cdot\modc u\,d\modc a.\label{eq:real_load_potential}
\end{equation}
Finally, the \emph{total potential} is:
\begin{equation}
\tp^{r}(\modc u):=\eef^{r}(\modc u)-\lp^{r}(\modc u).\label{eq:real_potential}
\end{equation}
Hereafter we denote by $\modc u_{\min}^r$ the unique minimizer of $\tp^{r}$:
\begin{equation}
\modc u_{\min}^r=\arg\min_{\modc u\in\mathcal{U}_{r}}\tp^{r}(\modc u).\label{eq:real_minimizer}
\end{equation}
\section{The Timoshenko Ansatz and its mechanical interpretation}\label{sec3}
Roughly speaking, according to the kinematical \emph{Ansatz} on which Timoshenko's
beam model is based, the beam's cross-section is regarded as rigid, while
the beam's axis may deform arbitrarily.  In fact, a prototypical
Timoshenko displacement field has the following form:
\begin{equation}
u_T=u_{1}(x_{3})e_{1}+u_{2}(x_{3})e_{2}+\left(u_{3}(x_{3})+x_{2}\psi_1(x_{3})-x_{1}\psi_2(x_{3})\right)e_{3},\label{eq:Timoshenko_dis}
\end{equation}
where $u_{i}$ and $\psi_{\alpha}$ are real-valued functions defined on $\left(0,L\right)$.

It was shown in \cite{Lembo2001} that, to within rigid global displacements, the above displacement field
is the general solution of the following set of PDEs:
\begin{equation}\label{TimoPDE}
\begin{split}
2(Eu)_{\alpha\beta}=u_{\alpha,\beta}+u_{\beta,\alpha} & =0,\\
2\big((Eu)_{3\alpha}\big),_\beta=u_{3,\alpha\beta}+u_{\alpha,3\beta} & =0.
\end{split}
\end{equation}
The following precise kinematical interpretation of \eqref{eq:Timoshenko_dis}  was given: for each fixed axial coordinate $x_3$,  the first-order PDEs in \eqref{TimoPDE} imply that  cross-section fibers neither change their length nor change their mutual angle, while the second-order PDEs imply that the change in angle between an axial fiber and a cross-section fiber does not depend on the direction of the latter.
It was also remarked in \cite{Lembo2001} that the each of the above PDEs can be interpreted as a non-dissipative internal-constraint condition, and the reaction stresses  
and hyperstresses maintaining those constraints were determined. 

In \cite{Miara2007}, it was suggested that both shearable-structure theories, Reissner-Mindlin's for plates and Timoshenko's for beams, are derivable from a three-dimensional elastic-energy functional that  includes a second-order strain contribution. 
This suggestion will be taken into account in the construction of the energy sequence given here below.

\section{A sequence of variational problems }\label{sec:A-sequence-of_domains}

In this section we construct a sequence of variational problems, parameterized by  $\eps$,  such that the ``real problem'' presented in Section \ref{sec:The_real_problem} is the element of the sequence whose parameter $\eps$ is equal to $\eps_r$. Our construction of the typical problem in this sequence is achieved in a number of scaling steps, both for data and candidate solutions.

\subsection{Domain} 
For $\alpha$ positive, let $R^\alpha:=\diag(\alpha,\alpha,1)$ be a diagonal $3\times 3$ matrix; moreover, for $x^r\equiv(x^r_1,x^r_2,x^r_3)$ a typical point of $\Omega_r$ and for any given $\eps>0$, let $\Omega$ and $\Omega_\eps$ denote the sets in $\reals^3$ whose typical points are, respectively, $x\equiv(x_1,x_2,x_3)=R^{1/\eps_r}x^r$ and $x^\eps\equiv (\eps x_1,\eps x_2,x_3)=R^\eps x$ (note that, consequently, $x^\eps=R^{\eps/\eps_r} x^r$). It follows from these definitions that
$$\Omega_{\eps}=\omega_{\eps}\times(0,L),\quad\textrm{where}\quad \omega_{\eps}:=\frac{\eps}{\eps_{r}}\omega_{r},$$
and that
$$\Omega=\omega\times(0,L),\quad\textrm{where}\quad \omega=\omega_1:=\frac{1}{\eps_{r}}\omega_{r}.$$
Therefore, $\omega_{\eps}$ and $\omega$ are nothing but the domains in $\reals^2$ obtained by homothetic rescaling of $\omega_{r}$ by, respectively, factors
$\eps/\eps_{r}$ and $1/\eps_r$.
Note for later use that the following relationships hold:
\begin{itemize}
\item $dv^\eps=\eps^2\,dv$, between the volume measures of $\Omega_{\eps}$ and $\Omega$;
\item $da^\eps=\eps^2\,da$, between the area measures of $\omega_{\eps}$ and $\omega$. 
\end{itemize}
%
%
\subsection{Displacement and strain fields} Given a displacement field $\modc u^{\eps}:\Omega_{\eps}\to\reals^{3}$,
we let the \emph{scaled displacement} $u^{\eps}:\Omega\to\reals^{3}$ be defined by
\begin{equation}\label{scadisp}
u^{\eps}(x):=R^{\eps}\,\modc u^{\eps}(R^\eps x);
\end{equation}
in Cartesian components, 
$$
u^\eps_1=\eps\, \modc u^{\eps}_1\circ R^\eps, \quad u^\eps_2=\eps\, \modc u^{\eps}_2\circ R^\eps,\quad u^\eps_3= \modc u^{\eps}_3\circ R^\eps.
$$
It
follows form definition \eqref{scadisp} that
\begin{equation}
\nabla u^{\eps}= R^{\eps}(\nabla\modc u^{\eps})R^{\eps},\label{eq:mod_grad_dis}
\end{equation}
and hence that
\begin{equation}
Eu^{\eps}=R^{\eps}(E\modc u^{\eps})R^{\eps};\label{eq:strain_coardinate_transformation}
\end{equation}
we call 
\begin{equation}
E^\eps u^{\eps}:=\left(R^{\eps}\right)^{-1}Eu^{\eps}\left(R^{\eps}\right)^{-1}\label{eq:str_eps(u_eps)}
\end{equation}
 the \emph{scaled strain}, and we record here its component form:
\begin{equation}
E^\eps u^{\eps}=
{\displaystyle \left(\begin{array}{ccc}
\displaystyle\frac{\left(\str u^{\eps}\right)_{11}}{\eps^{2}} & \displaystyle\frac{\left(\str u^{\eps}\right)_{12}}{\eps^{2}} &\displaystyle \frac{\left(\str u^{\eps}\right)_{13}}{\eps}\\[5pt]
\displaystyle\frac{\left(\str u^{\eps}\right)_{21}}{\eps^{2}} &\displaystyle \frac{\left(\str u^{\eps}\right)_{22}}{\eps^{2}} &\displaystyle \frac{\left(\str u^{\eps}\right)_{23}}{\eps}\\[5pt]
\displaystyle\frac{\left(\str u^{\eps}\right)_{31}}{\eps} & \displaystyle\frac{\left(\str u^{\eps}\right)_{32}}{\eps} &\displaystyle \left(Eu^{\eps}\right)_{33}
\end{array}\right)}.\label{eq:E_eps(u_eps)}
\end{equation}
\subsection{Elastic and load potentials, total potential}\quad{}\\

\noindent
(i) Let
\begin{equation}
\begin{split}\W^{\eps}(\str) & :=\frac{1}{2}\big[2\mu E_{\alpha\beta}E_{\alpha\beta}+\lambda\left(E_{\alpha\alpha}\right)^{2}+2\tau_{2}E_{33}E_{\alpha\alpha}\\
 & \hspace{3cm}+4\gamma\left(\frac{\eps}{\eps_{r}}\right)^{2}E_{\alpha3}E_{\alpha3}+\tau_{1}E_{33}^{2}\big]
\end{split}
\label{eq:W_eps}
\end{equation}
(note that $\W^{\eps_{r}}(\str)=\W(E)$). The \emph{elastic potential} is the functional 
%
\begin{equation}
\begin{split}
\eef^{\eps}\left(u^{\eps}\right) &:=\int_{\Omega}\W^{\eps}\left(\str^{\eps}u^{\eps}\right)dv\\ & \hspace{1.5cm}+\frac{1}{2}\tau_{R}\int_{\Omega}\left(\frac{\eps-\eps_{r}}{\eps}\right)^{2}\sum_{\alpha,\beta}\left(u_{3,\alpha\beta}^{\eps}+u_{\alpha,3\beta}^{\eps}\right)^{2}dv,\;\;\tau_{R}>0.\label{eq:els_eng_eps}
\end{split}
\end{equation}
defined over the collection of all elements  $v\in H^{1}\left(\Omega;\reals^{3}\right)$ such that $(v_{3,\alpha\beta}+v_{\alpha,3\beta})\in L^{2}\left(\Omega\right)$ (cf. definition \eqref{eq:admisible_disp_(A)}).
%
%


\noindent (ii) For $\modc b^{r}:\Omega_{r}\to\reals^{3}$ and $\modc c^{r}:\bnd_{N}\Omega_{r}\to\reals^{3}$,
the real distance and contact loads introduced in Section \ref{sec:The_real_problem}, we let
\[
\begin{split}b^{r} (x) & :=\left(R^{\eps_{r}}\right)^{-1}\modc b^{r} (R^{\eps_r}x)
\end{split}
\]
and
\[
c^{r}(x)=\begin{cases}
\left(R^{\eps_{r}}\right)^{-1}\modc c^{r}(R^{\eps_r}x) & \;\text{on }\omega\times\{L\},\\
\frac{1}{\eps_{r}}\left(R^{\eps_{r}}\right)^{-1}\modc c^{r}(R^{\eps_r}x) & \;\text{on }\bnd\omega\times(0,L)
\end{cases}
\]
be the \emph{scaled loads}, and we assume that $b^{r}\in L^{2}\left(\Omega;\reals^{3}\right)$
and $c^{r}\in L^{2}\left(\bnd_{N}\Omega;\reals^{3}\right)$. The \emph{load
potential} is the functional $\lp:H^{1}\left(\Omega;\reals^{3}\right)\to\reals$ defined by
\begin{equation}
\lp\left(u^{\eps}\right):=\int_{\Omega}b^{r}\cdot u^{\eps}dv+\int_{\bnd_{N}\Omega}c^{r}\cdot u^{\eps}da.\label{eq:Load_potential_eps}
\end{equation}
%
(iii) the \emph{total potential} is 
\begin{equation}
\tp^{\eps}\left(u^{\eps}\right):=\eef^{\eps}\left(u^{\eps}\right)-\lp\left(u^{\eps}\right).\label{eq:total_potential_eps}
\end{equation}

\smallskip
\noindent
{\bf Remark 1.}
 We close this section by observing that energies  $\tp^{r}$  and $\tp^{\eps_{r}}$ coincide, to within a multiplicative constant. Precisely, with the use of the above data and solution scalings, it can be shown that

\begin{equation}\label{Pir}
\tp^{r}\left(\modc u^{\eps_{r}}\right) =\eps_{r}^{2}\,\tp^{\eps_{r}}\left(u^{\eps_{r}}\right).
\end{equation}
Since the multiplicative constant $\eps_r^2$ does not affect the minimization process, we deduce that, up to a change of variables, the minimizers of $\tp^{\eps_{r}}$ and $\tp^{r}$ coincide.

\section{Compactness\label{sec:Compactness}}

In this Section we present a compactness lemma for a sequence $\left\{u^{\eps}\right\}$
with uniformly bounded energy. The functional $\tp^{\eps}$ is well-defined on 
\begin{equation}
\ad:=\left\{ v\in H^{1}\left(\Omega;\reals^{3}\right)\mid v=0\mbox{ on } \omega\times \{x_{3}=0\}\;\text{and}\; (v_{3,\alpha\beta}+v_{\alpha,3\beta})\in L^{2}\left(\Omega\right)\right\};\label{eq:admisible_disp_(A)}
\end{equation}
we extend $\tp^{\eps}$ to all of $L^{2}\left(\Omega;\reals^{3}\right)$, without renaming it, as follows:
\[
\tp^{\eps}\left(u\right):=\begin{cases}
\tp^{\eps}\left(u\right) & \textrm{for}\; u\in\ad,\\
\infty & \textrm{for}\; u\in L^2(\Omega;\reals^3)\setminus \ad.
\end{cases}
\]
Moreover, for
\[
H_{D}^{1}\left(0,L\right)=\left\{ v\in H^{1}\left(0,L\right)\mid v(0)=0\right\} ,
\]
we let 
\begin{equation}
\begin{split}\TD:= & \left\{ u\in H^{1}\left(\Omega;\reals^{3}\right)\mid\exists\; u_{i}^{0},\,\psi_{\alpha}\in H_{D}^{1}\left(0,L\right),\right.\\
 &\quad \left.u=u_{1}^{0}(x_{3})e_{1}+u_{2}^{0}(x_{3})e_{2}+\left(u_{3}^{0}(x_{3})+x_{2}\,\psi_1(x_{3})-x_{1}\,\psi_2(x_{3})\right)e_{3}\right\} ,
\end{split}
\label{eq:admisible_disp_Timo}
\end{equation}
We view $\ad$ as the set of admissible displacements
of a beam-like body, and $\TD\subset\ad$ as the subset of all
displacements compatible with Timoshenko's kinematic \emph{Ansatz}. 
\begin{lem}[Compactness lemma]
\label{lem:compactnes} Let a given sequence $\left\{ u^{\eps}\right\} \subset\ad$ be
such that 
\begin{equation}
\sup_{\eps}\tp^{\eps}\left(u^{\eps}\right)<\infty.\label{eq:bounded_eef}
\end{equation}
Then, there are a subsequence $\left\{ u^{\eps}\right\} $ (not
relabeled) and an element $u$ of $\TD$ such that $u^{\eps}\rightharpoonup u$ in $H^{1}\left(\Omega;\reals^{3}\right)$. Moreover, 
\[
\left(Eu^{\eps}\right)_{\alpha\beta}\to0\;\;\text{in}\: L^{2}\left(\Omega\right).
\]
\end{lem}
\begin{proof}
By means of H\"{o}lder's inequality,  (\ref{eq:material_bounds}),
the trace theorem for $H^{1}$-functions, and Korn's inequality, which
holds because $u^{\eps}=0$ on $\omega\times \{x_{3}=0\}$, we find: 
\begin{eqnarray*}
\tp^{\eps}\left(u^{\eps}\right) \!\!\!\!& \geq &\!\!\!\! C\sum_{\alpha,\beta}\int_{\Omega}\left(\left(\str^{\eps}u^{\eps}\right)_{\alpha\beta}^{2}+\eps^{2}\left(\str^{\eps}u^{\eps}\right)_{\alpha3}^{2}+\left(\str^{\eps}u^{\eps}\right)_{33}^{2}+\frac{1}{\eps^{2}}\left(u_{3,\alpha\beta}^{\eps}+u_{\alpha,3\beta}^{\eps}\right)^{2}\right)dx\\
 &  & \hspace{3cm}-\left\Vert b^{r}\right\Vert _{L^{2}(\Omega)}\left\Vert u^{\eps}\right\Vert _{L^{2}(\Omega)}-\left\Vert c^{r}\right\Vert _{L^{2}\left(\bnd_{D}\Omega\right)}\left\Vert u^{\eps}\right\Vert _{L^{2}\left(\bnd_{D}\Omega\right)},\\
 & \geq &\!\!\!\! C_{1}\left\Vert \str u^{\eps}\right\Vert _{L^{2}(\Omega)}^{2}-C_{2}\left\Vert u^{\eps}\right\Vert _{H^{1}(\Omega)},\\
 & \geq &\!\!\!\! C_{3}\left\Vert u^{\eps}\right\Vert _{H^{1}(\Omega)}^{2}-C_{2}\left\Vert u^{\eps}\right\Vert _{H^{1}(\Omega)},\\
 & \geq &\!\!\!\! C_{4}\left\Vert u^{\eps}\right\Vert _{H^{1}(\Omega)}^{2}-C_{5}.
\end{eqnarray*}
It follows from assumption (\ref{eq:bounded_eef}) that $\left\{ u^{\eps}\right\} $
is a bounded sequence in $H^{1}\left(\Omega;\reals^{3}\right)$. Moreover, from the first line of the inequalities above we deduce that all sequences

\[
\left\{\frac{\left(Eu^{\eps}\right)_{\alpha\beta}}{\eps^{2}}\right\},\;\left\{\left(Eu^{\eps}\right)_{\alpha3}\right\},\;\left\{\left(Eu^{\eps}\right)_{33}\right\},\;\left\{\frac{u_{3,\alpha\beta}^{\eps}+u_{\alpha,3\beta}^{\eps}}{\eps}\right\},
\]
are bounded in $L^{2}(\Omega)$. Hence, up to a subsequence,
$u^{\eps}\rightharpoonup u$ in $H^{1}\left(\Omega;\reals^{3}\right)$
for some $u\in H^{1}\left(\Omega;\reals^{3}\right)$. In addition, from the fact that $\left\{ \frac{\left(Eu^{\eps}\right)_{\alpha\beta}}{\eps^{2}}\right\} $
is a bounded sequence in $L^{2}\left(\Omega\right)$, it follows that
$u_{\alpha,\beta}^{\eps}+u_{\beta,\alpha}^{\eps}=2\left(Eu^{\eps}\right)_{\alpha\beta}\to0$
in $L^{2}\left(\Omega\right)$; and,  $u_{\alpha,\beta}^{\eps}+u_{\beta,\alpha}^{\eps}\rightharpoonup u_{\alpha,\beta}+u_{\beta,\alpha}$
in $L^{2}\left(\Omega\right)$; we conclude that $u_{\alpha,\beta}+u_{\beta,\alpha}=0$. Finally, with a similar argument, we can show that $u_{3,\alpha\beta}^{\eps}+u_{\alpha,3\beta}^{\eps}\rightharpoonup u_{3,\alpha\beta}+u_{\alpha,3\beta}=0$
in $L^{2}\left(\Omega\right)$. Thence, $u\in\TD$. 
\end{proof}

\section{The $\Gamma$-limit of $\{\protect\tp^{\protect\eps}\}$\label{sec:The-G-limit}}

In this section we analyze the $\Gamma$-limit of the sequence of
functionals $\tp^{\eps}$. For 
\begin{equation}
\begin{split}\W_{\tau}^{\eps}\left(E_{i3}\right) & :=\min_{g_{\alpha\beta}}\big\{ \W^{\eps}\big(\sum_{\alpha}\str_{\alpha3}\left(e_{\alpha}\otimes e_{3}+e_{3}\otimes e_{\alpha}\right)+E_{33}e_{3}\otimes e_{3}\\
 & \hspace{4,5cm}+\sum_{a,\beta}\frac{g_{\alpha\beta}}{2}\left[e_{\alpha}\otimes e_{\beta}+e_{\beta}\otimes e_{\alpha}\right]\big)\big\} ,
\end{split}
\label{eq:W_tau_eps}
\end{equation}
a simple calculation shows that 
\begin{equation}
\W_{\tau}^{\eps}\left(E_{i3}\right)=\frac{1}{2}[4\gamma(\frac{\eps}{\eps_{r}})^{2}(E_{13}^{2}+E_{23}^{2})+(\tau_{1}-\frac{\tau_{2}^{2}}{\lambda+\mu})E_{33}^{2}].\label{eq:W_eps_tau}
\end{equation}
Define 
\begin{equation}
\W_{\tau}(E_{i3}):=\W_{\tau}^{1}(E_{i3})=\frac{1}{2}[4\gamma\,\eps_{r}^{-2}(E_{13}^{2}+E_{23}^{2})+(\tau_{1}-\frac{\tau_{2}^{2}}{\lambda+\mu})E_{33}^{2}],\label{eq:W_tau_eps_tau}
\end{equation}
so that 
\begin{equation}
\W_{\tau}^{\eps}(E_{i3})=\W_{\tau}(\eps E_{13},\eps E_{23},E_{33}).\label{eq:W_eps_tau_tau}
\end{equation}
Define $\eef_{\tau}:L^{2}\left(\Omega;\reals^{3}\right)\to\reals$
by 
\begin{equation}
\eef_{\tau}(u):=\begin{cases}
\int_{\Omega}\W_{\tau}\left(\left(Eu\right)_{i3}\right)dx, & u\in\TD\\
\infty, & \text{otherwise }
\end{cases},\label{eq:en_W_eps_tau_tau}
\end{equation}
and $\tp:L^{2}\left(\Omega;\reals^{3}\right)\to\reals$ by 
\begin{equation}\label{pottot}
\tp\left(u\right):=\begin{cases}
\eef_{\tau}\left(u\right)-\lp\left(u\right), & u\in\TD\\
\infty, & \text{otherwise}
\end{cases}.
\end{equation}

\begin{thm}
The sequence $\{\tp^{\eps}\}$ $\Gamma$-converges to $\Pi$ in the $L^{2}\left(\Omega;\reals^{3}\right)$-topology,
that is to say,\label{thm:gamma_conver}

$(i)$~(liminf~inequality) for every sequence $\left\{ u^{\eps}\right\} \subset L^{2}\left(\Omega;\reals^{3}\right)$
and for every $u\in L^{2}\left(\Omega;\reals^{3}\right)$ such that $u^{\eps}\to u$ in $L^{2}\left(\Omega;\reals^{3}\right)$,
\begin{equation}
\liminf_{\eps\to0}\tp^{\eps}(u^{\eps})\geq\tp(u);\label{eq:lim_inf_in}
\end{equation}

$(ii)$~(recovery~sequence) for every $u\in L^{2}\left(\Omega;\reals^{3}\right)$,
there is a sequence $\left\{ u^{\eps}\right\} \subset L^{2}\left(\Omega;\reals^{3}\right)$
such that $u^{\eps}\to u$ in $L^{2}\left(\Omega;\reals^{3}\right)$ and 
\begin{equation}
\limsup_{\eps\to0}\tp^{\eps}(u^{\eps})\leq\tp(u).\label{eq:lim_sup_in}
\end{equation}
\end{thm}
\begin{proof}
We start by proving $(i)$. Let $\left\{ u^{\eps}\right\} \subset L^{2}\left(\Omega;\reals^{3}\right)$
and $u\in L^{2}\left(\Omega;\reals^{3}\right)$ such that $u^{\eps}\to u$
in $L^{2}\left(\Omega;\reals^{3}\right)$. The inequality in 
(\ref{eq:lim_inf_in}) is nontrivial only for 
\[
\liminf_{\eps\to0}\tp^{\eps}\left(u^{\eps}\right)<\infty.
\]
Thus, (by passing, if needed, to a subsequence) we may assume that
\[
\lim_{\eps\to0}\tp^{\eps}\left(u^{\eps}\right)=\liminf_{\eps\to0}\tp^{\eps}\left(u^{\eps}\right)<\infty.
\]
By Lemma \ref{lem:compactnes}, the sequence $\left\{ u^{\eps}\right\} $
converges weakly in $H^{1}\left(\Omega;\reals^{3}\right)$ to an element
$u\in\TD$. Thus, 
\begin{eqnarray*}
\liminf_{\eps\to0}\tp^{\eps}(u^{\eps}) & \geq & \liminf_{\eps\to0}\{ \int_{\Omega}\W^{\eps}\left(\str^{\eps}u^{\eps}\right)dx-\lp\left(u^{\eps}\right)\} ,\\
 & \geq & \liminf_{\eps\to0}\{ \int_{\Omega}\W_{\tau}^{\eps}\left(\left(\str^{\eps}u^{\eps}\right)_{i3}\right)dx-\lp\left(u^{\eps}\right)\} ,\\
 & = & \liminf_{\eps\to0}\{ \int_{\Omega}\W_{\tau}\left(\eps\left(\str^{\eps}u^{\eps}\right)_{13},\eps\left(\str^{\eps}u^{\eps}\right)_{23},\left(\str^{\eps}u^{\eps}\right)_{33}\right)dx-\lp\left(u^{\eps}\right)\} ,\\
 & = & \liminf_{\eps\to0}\{ \int_{\Omega}\W_{\tau}\left(\left(\str u^{\eps}\right)_{13},\left(\str u^{\eps}\right)_{23},\left(\str u^{\eps}\right)_{33}\right)dx-\lp\left(u^{\eps}\right)\} ,\\
 & \geq & \int_{\Omega}\W_{\tau}\left(\left(\str u\right)_{13},\left(\str u\right)_{23},\left(\str u\right)_{33}\right)dx-\lp\left(u\right),\\
 & = & \eef_{\tau}(u)-\lp\left(u\right)=\tp(u).
\end{eqnarray*}
We point out that: in the first inequality we dispensed with second-order terms;
the second inequality makes use of  (\ref{eq:W_tau_eps});
in the third line we applied  (\ref{eq:W_eps_tau_tau}); the
last inequality follows from the convexity of $\W_{\tau}$ and the
continuity of $\lp$. 

We now prove $(ii)$. Let $u\in L^{2}\left(\Omega;\reals^{3}\right)$.
Since inequality (\ref{eq:lim_sup_in}) is trivial for  $u\notin\TD$, 
we only need consider the case when $u\in\TD$, and hence has the representation specified in \eqref{eq:admisible_disp_Timo}.

At first, we restrict attention to functions $u_{i}^{0},\: \psi_{\alpha}$ belonging to $C^{\infty}(0,L)$
and equal to zero in a neighborhood of $x_3=0$; note that these functions form a
dense subset of $H_{D}^{1}\left(0,L\right)$. We consider the sequence whose typical term is:
\[
u^{\eps}:=u+\eps^{2}\hat{u},
\]
where $\hat{u}$ is defined by  
\begin{eqnarray}
\hat{u}_{1} & := & -\eta(x_{1}u_{3,3}^{0}+x_{1}x_{2}\psi_{1,3}-\frac{x_{1}^{2}}{2}\psi_{2,3}+\frac{x_{2}^{2}}{2}\psi_{2,3}),\nonumber \\
\hat{u}_{2} & := & -\eta(x_{2}u_{3,3}^{0}-x_{2}x_{1}\psi_{2,3}+\frac{x_{2}^{2}}{2}\psi_{1,3}-\frac{x_{1}^{2}}{2}\psi_{1,3}),\label{eq:Recovery_seq}\\
\hat{u}_{3} & := & 0,\nonumber 
\end{eqnarray}
with 
$\eta:=\tau_{2}/2\left(\mu+\lambda\right)$. A simple computation yields:
\begin{eqnarray*}
\left(\str^{\eps}u^{\eps}\right){}_{11} & = & \left(\str\hat{u}\right){}_{11}=-\eta\left(\str u\right){}_{33},\\
\left(\str^{\eps}u^{\eps}\right)_{12} & = & 0,\\
\left(\str^{\eps}u^{\eps}\right)_{22} & = & \left(\str\hat{u}\right){}_{22}=-\eta\left(\str u\right){}_{33},\\
\left(\str^{\eps}u^{\eps}\right)_{\alpha3} & = & \left(\str u\right){}_{\alpha3}+\eps^{2}\left(\str\hat{u}\right){}_{\alpha3},\\
\left(\str^{\eps}u^{\eps}\right)_{33} & = & \left(\str u\right){}_{33}.
\end{eqnarray*}
Thus, 
\begin{eqnarray*}
\eef^{\eps}\left(u^{\eps}\right) \!\!\!\!& = &\!\!\!\! \int_{\Omega}\frac{1}{2}\big[4\mu\eta^{2}\left(\str u\right){}_{33}^{2}+4\lambda\eta^{2}\left(\str u\right){}_{33}^{2}+2\tau_{2}\left(Eu\right){}_{33}\left(-2\eta\left(\str u\right){}_{33}\right)\\
 &  & \hspace{2,5cm}+\tau_{1}\left(Eu\right){}_{33}^{2}+(\frac{\eps}{\eps_{r}})^{2}4\gamma(\frac{\left(\str u\right){}_{\alpha3}+\eps^{2}(\str\hat{u}){}_{\alpha3}}{\eps})^{2}\\
 &  & \hspace{4cm}+\sum_{\alpha,\beta}(\frac{\eps_{r}-\eps}{\eps})^{2}\eps^{4}(\hat{u}_{3,\alpha\beta}+\hat{u}_{\alpha,3\beta})^{2}\big]dx.
\end{eqnarray*}
Upon rearranging. we obtain:
\begin{eqnarray*}
\eef^{\eps}\left(u^{\eps}\right) \!\!\!\! & = & \!\!\!\! \int_{\Omega}\frac{1}{2}\big[\big(\tau_{1}-\frac{\tau_{2}^{2}}{\mu+\lambda}\big)\left(Eu\right){}_{33}^{2}+\big(\frac{1}{\eps_{r}}\big)^{2}4\gamma\left(\left(\str u\right){}_{\alpha3}+\eps^{2}\left(\str\hat{u}\right){}_{\alpha3}\right)^{2}\\
 &  & \hspace{4cm}+\sum_{\alpha,\beta}\left(\eps_{r}-\eps\right)^{2}\eps^{2}\left(\hat{u}_{3,\alpha\beta}+\hat{u}_{\alpha,3\beta}\right)^{2}\big]\,dx.
\end{eqnarray*}
By passing to the limit, it follows that 
\[
\lim_{\eps\to0}\eef^{\eps}(u^{\eps})=\int_{\Omega}\frac{1}{2}[(\frac{1}{\eps_{r}})^{2}4\gamma\left(\str u\right){}_{\alpha3}^{2}+(\tau_{1}-\frac{\tau_{2}^{2}}{\mu+\lambda})(Eu)_{33}^{2}]dx=\eef_{\tau}(u);
\]
since
\[
\lim_{\eps\to0}\lp\left(u^{\eps}\right)=\lp(u),
\]
the proof of $(ii)$ is achieved, under the smoothness assumptions stated at the beginning of this paragraph. 

The general $u\in\TD$
case is handled by a standard diagonalization argument. Indeed, let $\{u_{k}\}\subset\TD\cap C^{\infty}\left(\Omega;\reals^{3}\right)$
such that $u_{k}\to u$ in $H^{1}\left(\Omega;\reals^{3}\right)$; moreover, let $\{u_{k}^{\eps}\}$ be the recovery sequence for $u_{k}$
as defined by (\ref{eq:Recovery_seq}). Since 
$$
\lim_{k\to\infty} \lim_{\eps\to0}\|u^\eps_k-u\|_{H^1(\Omega)}=0,
$$
and
\[
\lim_{k\to\infty} \lim_{\eps\to0}\tp^{\eps}\left(u_{k}^{\eps}\right) =\lim_{k\to\infty} \tp\left(u_{k}\right) =\tp\left(u\right)
\]
we can find an increasing map $\eps\to k_\eps$ such that $u^\eps_{k_\eps}\to u$ in $H^1(\Omega;\reals^3)$ and 
$$
\limsup_{\eps\to0}\tp^{\eps}\left(u_{k_\eps}^{\eps}\right)\le\lim_{k\to\infty} \lim_{\eps\to0}\tp^{\eps}\left(u_{k}^{\eps}\right) =\tp\left(u\right).
$$
\end{proof}
\begin{thm}\label{thm3}
Let $u_{\min}^{\eps}$ be the minimizer of $\tp^{\eps}$, and let $u_{\min}$
be the minimizer of $\tp$. Then,

$(i)$ the sequence $\tp^{\eps}\left(u_{\min}^{\eps}\right)$ converges
to $\tp\left(u_{\min}\right)$;

$(ii)$ the sequence $\left\{ u_{\min}^{\eps}\right\} $ converges
to $u_{\min}$ strongly in $H^{1}\left(\Omega;\reals^{3}\right)$.\end{thm}
\begin{proof}
As to the first claim, given that $\sup_{\eps}\left\{ \tp^{\eps}(u_{\min}^{\eps})\right\} <\infty$,
it follows from Lemma \ref{lem:compactnes} that $\left\{ u_{\min}^{\eps}\right\} $
converge weakly in $H^{1}\left(\Omega;\reals^{3}\right)$ (up to a
subsequence) to $u_{\lim}\in\TD$. For a general $u\in L^{2}\left(\Omega;\reals^{3}\right)$,
by $(ii)$ of Theorem \ref{thm:gamma_conver}, there exists a recovery
sequence $\left\{ u^{\eps}\right\} $ such that 
\[
\tp(u)\geq\limsup_{\eps\to0}\tp^{\eps}\left(u^{\eps}\right)\geq\limsup_{\eps\to0}\tp^{\eps}\left(u_{\min}^{\eps}\right).
\]
From $(i)$ of Theorem \ref{thm:gamma_conver}, we deduce that 
\begin{equation}
\tp(u)\geq\limsup_{\eps\to0}\tp^{\eps}\left(u_{\min}^{\eps}\right)\geq\liminf_{\eps\to0}\tp^{\eps}\left(u_{\min}^{\eps}\right)\geq\tp\left(u_{\lim}\right).\label{eq:chain_2}
\end{equation}
Now, $u$ can be chosen arbitrarily; we  take $u=u_{\min}$, and obtain $\tp\left(u_{\min}\right)\geq\tp\left(u_{\lim}\right)$; as $\tp\left(u_{\min}\right)\leq\tp\left(u_{\lim}\right)$,
 part $(i)$ of the theorem
 follows and $u^{\varepsilon}_{\min} \rightharpoonup u_{\min} $ in $H^1\left(\Omega;\reals^{3}\right)$. 

As to the second claim, by means of \eqref{eq:els_eng_eps}, (\ref{eq:W_tau_eps}), and (\ref{eq:W_eps_tau_tau}),
we deduce that 
\begin{equation}
\begin{split}\eef^{\eps}\left(u_{\min}^{\eps}\right)-\eef_{\tau}\left(u_{\min}\right) & \geq\int_{\Omega}\left[\W^{\eps}\left(\str^{\eps}u_{\min}^{\eps}\right)-\W_{\tau}\left(\left(\str u_{\min}\right)_{i3}\right)\right]dx,\\
 & =\int_{\Omega}\left[\W^{\eps}\left(\str^{\eps}u_{\min}^{\eps}\right)-\W_{\tau}^{\eps}\left(\left(\str^{\eps}u_{\min}\right)_{i3}\right)\right]dx,\\
 & \geq\int_{\Omega}\left[\W_{\tau}^{\eps}\left(\left(\str^{\eps}u_{\min}^{\eps}\right)_{i3}\right)-\W_{\tau}^{\eps}\left(\left(\str^{\eps}u_{\min}\right)_{i3}\right)\right]dx.
\end{split}
\label{eq:W_eps-W_tau}
\end{equation}
On recalling the form of $\W_{\tau}^{\eps}$ given in (\ref{eq:W_eps_tau}),
it follows after some computations, with the aid of (\ref{eq:material_bounds}),
that there exists a constant $C>0$ for which
\[
\begin{split} & \int_{\Omega}\left[\W_{\tau}^{\eps}\left(\left(\str^{\eps}u_{\min}^{\eps}\right)_{i3}\right)-\W_{\tau}^{\eps}\left(\left(\str^{\eps}u_{\min}\right)_{i3}\right)\right]dx,\\
 & \geq C\sum_{i}\left(\left\Vert \left(\str u_{\min}^{\eps}\right)_{i3}-\left(\str u_{\min}\right)_{i3}\right\Vert _{L^{2}}^{2}+\int_{\Omega}\left[\left(\str u_{\min}^{\eps}\right)_{i3}-\left(\str u_{\min}\right)_{i3}\right]\left(\str u_{\min}\right)_{i3}dx\right).
\end{split}
\]
Combination of this inequality with  (\ref{eq:W_eps-W_tau}) yields:
\[
\begin{split}\tp^{\eps}\left(u_{\min}^{\eps}\right)-\tp\left(u_{\min}\right) & \geq C\sum_{i}\left(\left\Vert \left(\str u_{\min}^{\eps}\right)_{i3}-\left(\str u_{\min}\right)_{i3}\right\Vert _{L^{2}}^{2}\right.\\
 & \quad+\left.\int_{\Omega}\left[\left(\str u_{\min}^{\eps}\right)_{i3}-\left(\str u_{\min}\right)_{i3}\right]\left(\str u_{\min}\right)_{i3}dx\right)\\
 & \quad+\lp\left(u_{\min}^{\eps}\right)-\lp\left(u_{\min}\right).
\end{split}
\]
As $u_{\min}^{\eps}\rightharpoonup u_{\min}$ in $H^{1}\left(\Omega;\reals^{3}\right)$,
we deduce that 
\[
\int_{\Omega}\left[\left(\str u_{\min}^{\eps}\right)_{i3}-\left(\str u_{\min}\right)_{i3}\right]\left(\str u_{\min}\right)_{i3}dx\to0,
\]
thence, by part $(i)$ of the theorem, we find
\[
0\geq\lim\sup_{\eps\to0}\sum_{i}\left\Vert \left(\str u_{\min}^{\eps}\right)_{i3}-\left(\str u_{\min}\right)_{i3}\right\Vert _{L^{2}}^{2}.
\]
Thus, $\left(\str u_{\min}^{\eps}\right)_{i3}\to\left(\str u_{\min}\right)_{i3}$
in $L^{2}\left(\Omega\right)$. As the sequence $\left\{ \tp^{\eps}\left(u_{\min}^{\eps}\right)\right\} $
is bounded, it follows from Lemma \ref{lem:compactnes} that $\left(Eu_{\min}^{\eps}\right)_{\alpha\beta}\to0=\left(Eu_{\min}\right)_{\alpha\beta}$
in $L^{2}\left(\Omega\right)$; hence, $\str u_{\min}^{\eps}\to\str u_{\min}$
in $L^{2}\left(\Omega;\reals^{3\times3}\right)$. An application of
Korn's inequality concludes the proof.
\end{proof}

\section{The $\Gamma$-limit potential in terms of limit displacements}
%
Once again, recall that, for every $u\in \TD$, there are $u_{i}^{0},\, \psi_{\alpha}\in H_{D}^{1}\left(0,L\right)$ such that
 \begin{equation}\label{utdf}
 u=u_{1}^{0}(x_{3})e_{1}+u_{2}^{0}(x_{3})e_{2}+\left(u_{3}^{0}(x_{3})+x_{2}\psi_1(x_{3})-x_{1}\psi_2(x_{3})\right)e_{3}
 \end{equation}
 (cf. \eqref{eq:admisible_disp_Timo}). It follows that the non-null associated strain components are:
$$
(Eu)_{13}=\frac 12 (u_{1,3}^{0}-\psi_2), \quad
(Eu)_{23}=\frac 12 (u_{2,3}^{0}+\psi_1), \quad
(Eu)_{33}=u_{3,3}^{0}+x_2\psi_{1,3}-x_1\psi_{2,3}.
$$
Thus, in view of \eqref{eq:W_tau_eps_tau}, the elastic potential \eqref{eq:en_W_eps_tau_tau} reads:
\begin{eqnarray*}
\eef_{\tau}(u)\!\!\!\!\!&=&\!\!\!\!\!\frac{1}{2}\int_0^L dx_3\int_\omega\Big( \gamma\,\eps_{r}^{-2}[(u_{1,3}^{0}-\psi_2)^{2}+(u_{2,3}^{0}+\psi_1)^{2}]\\
& & \hspace{2cm}+(\tau_{1}-\frac{\tau_{2}^{2}}{\lambda+\mu})(u_{3,3}^{0}+x_2\psi_{1,3}-x_1\psi_{2,3})^{2}\Big)da \,,
\end{eqnarray*}
whence, on choosing for the first two Cartesian axes the principal axes of inertia of the cross-section $\omega$, we arrive at the \emph{one-dimensional} elastic-energy functional of Timoshenko's beam theory:
\begin{equation}
\begin{aligned}\label{potel}
\eef_{\tau}(u)&=\frac{1}{2}\int_0^L\Big(\frac{\gamma A}{\eps_{r}^2}[(u_{1,3}^{0}-\psi_2)^{2}+(u_{2,3}^{0}+\psi_1)^{2}]\\
& +(\tau_{1}-\frac{\tau_{2}^{2}}{\lambda+\mu})[A(u_{3,3}^{0})^2+J_1\psi_{1,3}^2+J_2\psi_{2,3}^{2}]\Big)dx_3,
\end{aligned} 
\end{equation}
where
$$
A:=\int_\omega \,da, \quad J_1:=\int_\omega x_2^2 \,da, \quad J_2:=\int_\omega x_1^2 \,da,
$$
denote, respectively, the area and the moments of inertia with respect to the $x_1$ and $x_2$ axes.\footnote{Under the assumptions,
$$
\int_\omega x_\alpha \,da=0 \quad\mbox{and}\quad\int_\omega x_1 x_2 \,da=0.
$$ 
}

It also follows from \eqref{utdf} that the load potential \eqref{eq:Load_potential_eps} takes the form:
\begin{eqnarray*}
\lp\left(u\right) \!\!\!&=&\!\!\!\int_{\Omega}\left(b^{r}_1 u_1^0+b^{r}_2 u_2^0+b^{r}_3 u_3^0+x_2b^r_3\psi_1-x_1b^r_3\psi_2\right)dv\\
& +&\int_{\bnd_{N}\Omega}\left(c^{r}_1 u_1^0+c^{r}_2 u_2^0+c^{r}_3 u_3^0+x_2c^r_3\psi_1-x_1c^r_3\psi_2\right)da,
\end{eqnarray*}
whence we deduce the \emph{one-dimensional} load functional of Timoshenko's beam theory:
\begin{equation}\label{potload}
\begin{aligned}
\lp\left(u\right)&=\int_0^L \left(f_1 u_1^0+f_2 u_2^0+f_3 u_3^0+m_1\psi_1+m_2\psi_2\right)dx_3\\
&  +F_1 u_1^0(L)+F_2 u_2^0(L)+F_3 u_3^0(L)+M_1\psi_1(L)+M_2\psi_2(L),
\end{aligned}
\end{equation}
where, for almost every $x_3\in (0,L)$, we have set
$$
f_i(x_3):=\int_{\omega}b^{r}_i(x_1,x_2,x_3) \,da+\int_{\partial\omega}c^{r}_i(x_1,x_2,x_3) \,ds,
$$

$$
m_1(x_3):=\int_{\omega}x_2b^{r}_3(x_1,x_2,x_3) \,da+\int_{\partial\omega}x_2c^{r}_3(x_1,x_2,x_3) \,ds,
$$
$$
m_2(x_3):=-\int_{\omega}x_1b^{r}_3(x_1,x_2,x_3) \,da-\int_{\partial\omega}x_1c^{r}_3(x_1,x_2,x_3) \,ds.
$$
and
$$
F_i:=\int_{\omega}c^{r}_i(x_1,x_2,L) \,da, \quad
M_1:=\int_{\omega}x_2c^{r}_3(x_1,x_2,L) \,da, \quad
M_2:=-\int_{\omega}x_1c^{r}_3(x_1,x_2,L)\,da\,.
$$

Interestingly, for the elastic potential as in \eqref{potel} and the load potential as in \eqref{potload}, the  total potential \eqref{pottot} may be decomposed 
into an \emph{axial-stretching} part $\tp_a$ and a \emph{bending} part $\tp_b$:
\begin{equation}\label{TPenergy}
\tp(u)=\tp_a(u_3^0)+\tp_b(u^0_\alpha, \psi_\alpha),
\end{equation}
where
\begin{equation}\label{TPenergyA}
\tp_a(u_3^0):=\frac{1}{2}(\tau_{1}-\frac{\tau_{2}^{2}}{\lambda+\mu})A\int_0^L(u_{3,3}^{0})^2\,dx_3-
\int_0^L f_3 u_3^0\,dx_3+F_3 u_3^0(L),
\end{equation}
and
\begin{eqnarray}
\tp_b(u^0_\alpha, \psi_\alpha)\!\!\!\!\!&:=&\!\!\!\!\!
\frac{1}{2}\frac{\gamma A}{\eps_{r}^2}\int_0^L(u_{1,3}^{0}-\psi_2)^{2}+(u_{2,3}^{0}+\psi_1)^{2}\,dx_3\nonumber\\
&+& \frac{1}{2}(\tau_{1}-\frac{\tau_{2}^{2}}{\lambda+\mu})\int_0^LJ_1\psi_{1,3}^2+J_2\psi_{2,3}^{2}\,dx_3
\label{TPenergyB}\\
&-&
\int_0^L f_1 u_1^0+f_2 u_2^0+m_1\psi_1+m_2\psi_2\,dx_3\nonumber\\
&- & F_1 u_1^0(L)-F_2 u_2^0(L)-M_1\psi_1(L)-M_2\psi_2(L).\nonumber
\end{eqnarray}
Thus, the minimization problem splits into two independent problems: the one for $u_3^0$, to be determined by minimizing the axial-stretching potential $\tp_a$, the other for $u^0_\alpha$ and $ \psi_\alpha$, to be determined
 by minimizing the bending potential $\tp_b$.

\medskip
\noindent
{\bf Remark 2.}
In the isotropic case, the elastic-energy density defined in \eqref{eq:W} only depends on the two parameters
$\lambda$ and $\mu$, because we have that
$$
\gamma=\mu,\quad\tau_1=\lambda+2\mu, \quad \tau_2=\lambda.
$$ 
Consequently, the elastic modulus appearing in both the axial-streching and the bending potentials reduces to the Young modulus of the material:
$$
\tau_{1}-\frac{\tau_{2}^{2}}{\lambda+\mu}=\frac{\mu(3\lambda+2\mu)}{\lambda+\mu}.
$$

\smallskip
\noindent
{\bf Remark 3.}
The  parameter $\eps_{r}$ enters both
 the three-dimensional limit elastic-energy density (\ref{eq:W_tau_eps_tau})
and the elastic part of the one-dimensional bending potential \eqref{TPenergyB}
(but not the elastic part of the axial-stretching potential \eqref{TPenergyA}). If, \emph{ceteris paribus}, we were to  let  $\eps_{r}\to0$ in \eqref{eq:W_tau_eps_tau}, we would achieve effortlessly a justification of the Bernoulli-Navier beam model, whose total potential obtains by letting  $\eps_{r}\to0$ in \eqref{TPenergyB}. In fact,
in the envisaged limit, the shear strains $E_{\alpha3}$ would be forced to converge to zero, which is tantamount to take $\psi_2=u_{1,3}^{0}$ and $\psi_1=-u_{2,3}^{0}$ in \eqref{utdf}.  This remark supports the engineer idea that  the 
Bernoulli-Navier model is fine for very slender beams, whereas the Timoshenko model is preferable whenever beams  are not so slender.

\section{Summary and conclusions}\label{concluding}
In Section \ref{sec:The_real_problem}, we have defined the total potential $\tp^r$ of a linearly elastic three-dimensional beam-like body; we have denoted the minimizer of this potential by 
$$
\modc u_{\min}^r=\arg\min_{\modc u\in\mathcal{U}_{r}}\tp^{r}(\modc u);
$$
and we have denoted the ratio between the diameter of the cross-section  and the length of the beam by $\eps_r$, a parameter that measures the slenderness of the beam-like body at hand. 
 In Section 3, we have recorded the form of the Timoshenko displacement field, and we have briefly discussed its mechanical significance.
In Section  \ref{sec:A-sequence-of_domains} we have constructed an $\eps$-sequence of functionals  $\tp^{\eps}$, such that $\tp^{\eps_r}$ is proportional to the ``real'' functional $\tp^{{r}}$ (see \eqref{Pir}).
After studying,  in Section \ref{sec:Compactness}, the compactness properties of sequences with bounded energy, we have identified the  $\Gamma$-limit of the sequence $\tp^\eps$ in Section \ref{sec:The-G-limit}. The $\Gamma$-limit $\tp$ turns out to be the total potential of a Timoshenko beam; in Section 7, we have written it in terms of the fields that parameterize the class of Timoshenko's displacements   (see \eqref{pottot} and \eqref{TPenergy}-\eqref{TPenergyB}), and we have shown that these parameter fields can be determined by solving two independent minimum problems for, respectively, axial stretching and bending.

In accordance with the notation introduced in Theorem \ref{thm3}, let
$$
u^\eps_{\min}=\arg\min_{u\in\mathcal{A}}\tp^{\eps}(u)\quad\textrm{and}\quad u_{\min}=\arg\min_{u\in\mathcal{TD}}\tp(u).
$$
The minimizer $u_{\min}$ is a Timoshenko-type displacement.
Since $\tp^r$ essentially coincides with $\tp^{\eps_r}$, we  deduce that
$\modc u_{\min}^r$ essentially coincides with $u^{\eps_r}_{\min}$. The exact relation between them follows immediately from \eqref{Pir} and \eqref{scadisp}, and is
\begin{equation}\label{minmin}
u^{\eps_r}_{\min}(x)=R^{\eps_r} \modc u_{\min}^r(R^{\eps_r}x).
\end{equation}
In Theorem \ref{thm3}, we have shown that
$$
u^\eps_{\min} \to u_{\min} \mbox{ in }H^1(\Omega;\reals^3);
$$
thus,  we can loosely say that, for small $\eps$, $u^\eps_{\min}$ is well approximated by $u_{\min}$.
In particular, for $\eps_r$ very small, $u^{\eps_r}_{\min}$ is well approximated by $u_{\min}$; we concisely write this as $u^{\eps_r}_{\min}\approx u_{\min}$.
By \eqref{minmin} we therefore find an approximation of the ``real'' displacement $\modc u_{\min}^r$, namely,
$$
\modc u_{\min}^r(x^r)\approx (R^{\eps_r})^{-1} u_{\min} (R^{1/\eps_r}x^r).
$$
This relation states that the ``real'' displacement $\modc u_{\min}^r$ is well approximated by
a Timoshenko-type displacement and that such an approximation can be constructed with the use of the minimizer of the $\Gamma$-limit we found.

\bigskip

\noindent \textbf{\textit{Acknowledgments.}} L. Falach and R. Paroni are grateful to Regione Autonoma della Sardegna for the support given under the project ``Modellazione Multiscala della Meccanica dei Materiali Compositi (M4C)''


\bibliographystyle{plain}

\end{document}